\documentclass[10pt, conference]{IEEEtran}
\usepackage[margin=1in,nohead,nofoot]{geometry}
\usepackage{subfigure}
\usepackage{epsfig,graphicx,psfrag}
\usepackage{amsfonts,amsmath,amssymb}

\newtheorem{definition}{Definition}
\newtheorem{theorem}{Theorem}[section]
\newtheorem{lemma}[theorem]{Lemma}

\begin{document}

\title{Communication Requirements for Generating Correlated Random Variables}

\author{
\authorblockN{Paul Cuff}
\authorblockA{Department of Electrical Engineering\\
Stanford University\\
E-mail: pcuff@stanford.edu }
}

\maketitle

\begin{abstract}
Two familiar notions of correlation are rediscovered as extreme operating points for simulating a discrete memoryless channel, in which a channel output is generated based only on a description of the channel input.  Wyner's ``common information'' coincides with the minimum description rate needed.  However, when common randomness independent of the input is available, the necessary description rate reduces to Shannon's mutual information.  This work characterizes the optimal tradeoff between the amount of common randomness used and the required rate of description.
\end{abstract}

\section{Introduction}
\label{section introduction}

What is the intrinsic connection between correlated random variables?  How much interaction is necessary to create correlation?

Many fruitful efforts have been made to quantify correlation between two random variables.  Each quantity is justified by the operational questions that it answers.  Covariance dictates the mean squared error in linear estimation.  Shannon's mutual information is the descriptive savings from side information in lossless source coding and the additional growth rate of wealth due to side information in investing.  G\'{a}cs and K\"{o}rner's common information \cite{Gacs} is the number of common random bits that can be extracted from correlated random variables.  It is less than mutual information.  Wyner's common information \cite{Wyner} is the number of common random bits needed to generate correlated random variables and is greater than mutual information.

This work provides a fresh look at two of these quantities --- mutual information and Wyner's common information (herein simply ``common information'').  Both are extreme points of the channel simulation problem, introduced as follows:  An observer ({\em encoder}) of an i.i.d. source $X_1,X_2,...$ describes the sequence to a distant random number generator ({\em decoder}) that produces $Y_1,Y_2,...$ (see Figure \ref{figure channel simulation}).  What is the minimum rate of description needed to achieve a joint distribution that is statistically indistinguishable (as measured by total variation) from the distribution induced by putting the source through a memoryless channel?

\begin{figure}
\psfrag{l1}[][][0.8]{$X^n$}
\psfrag{l2}[][][0.8]{$F_n$}
\psfrag{l3}[][][0.8]{$nR$ bits}
\psfrag{l4}[][][0.8]{$G_n$}
\psfrag{l5}[][][0.8]{$Y^n$}
\psfrag{l6}[][][1.0]{Channel Simulator:  $q(y|x)$}
\centering
\includegraphics[width=3.0in]{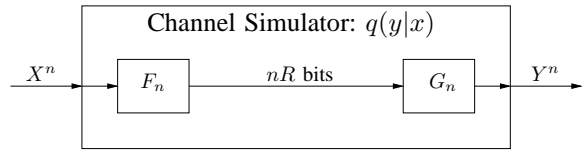}
\caption{A discrete-memoryless channel is simulated by two separate processors, $F$ and $G$.  The first processor, $F$, observes $X$ and the second processor, $G$, generates $Y$ after receiving a message at rate $R$ from $F$.  The minimum rate needed is the common entropy of $X$ and $Y$.}
\label{figure channel simulation}
\end{figure}

Channel simulation is a form of random number generation.  The variables $X^n$ come from an external source and $Y^n$ are generated to be correlated with $X^n$.  The channel simulation is successful if the total variation between the resulting distribution of $(X^n,Y^n)$ and the i.i.d. distribution that would result from passing $X^n$ through a memoryless channel is small.  This is a strong requirement.  It's stricter than the requirement that $(X^n,Y^n)$ be jointly typical as in the coordinated action work of Cover and Permuter \cite{Permuter}.  This total variation requirement means that any hypothesis test that a statistician comes up with to determine whether $X^n$ was passed through a real memoryless channel or the channel simulator will be virtually useless.

Wyner's result implies that in order to generate $X^n$ and $Y^n$ separately as an i.i.d. source pair they must share bits at a rate of at least the common information $C(X;Y)$ of the joint distribution.  In the channel simulation problem these shared bits come in the form of the description of $X^n$.\footnote{To achieve channel simulation with a rate as low as the common information one must change Wyner's relative entropy requirement in \cite{Wyner} to a total variation requirement as used in this work.}  However, the ``reverse Shannon theorem'' of Bennett and Shor \cite{Bennett} suggests that a description rate of the mutual information $I(X;Y)$ of the joint distribution is all that is needed to successfully simulate a channel.  How can we resolve this apparent contradiction?

The work of Bennett and Shor assumes that common random bits, or {\em common randomness}, independent of the source $X^n$ are available to the encoder and decoder.  In that setting, the common randomness provides a second connection between the source $X^n$ and output $Y^n$, in addition to the description of $X^n$.  Remarkably, even though it is independent from the source $X^n$, the common randomness assists in generating correlated random numbers and allows for description rates smaller than the common information $C(X;Y)$.

In this work, we characterize the tradeoff between the rate of available common randomness and the required description rate for simulating a discrete memoryless channel for a fixed input distribution, as in Figure \ref{figure channel simulation with common randomness}.  Indeed, the tradeoff region of Section \ref{section main result} confirms the two extreme cases.  If the encoder and decoder are provided with enough common randomness, sending $I(X;Y)$ bits per symbol suffices.  On the other hand, in the absence of common randomness one must spend $C(X;Y)$ bits per symbol.

\begin{figure}
\psfrag{l1}[][][0.8]{$X^n$}
\psfrag{l2}[][][0.8]{$F_n$}
\psfrag{l3}[][][0.8]{$nR_1$ bits}
\psfrag{l4}[][][0.8]{$G_n$}
\psfrag{l5}[][][0.8]{$Y^n$}
\psfrag{l6}[][][1.0]{Channel Simulator:  $q(y|x)$}
\psfrag{l7}[][][0.8]{$nR_2$ bits}
\centering
\includegraphics[width=3.0in]{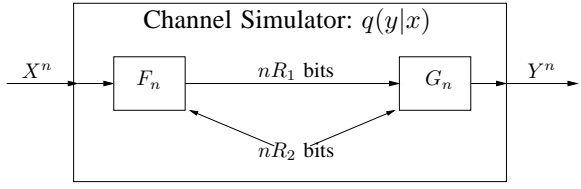}
\caption{A discrete-memoryless channel is simulated by two separate processors, $F$ and $G$.  The first processor, $F$, observes $X$ and common randomness independent of $X$ at rate $R_2$.  The second processor, $G$, generates $Y$ based on the common randomness and a message at rate $R_1$ from $F$.}
\label{figure channel simulation with common randomness}
\end{figure}

This result has implications in cooperative game theory, reminiscent of the framework investigated in \cite{Anantharam}.  Suppose a team shares the same payoff in a repeated game setting.  An opponent tries to anticipate and exploit patterns in the team's combined actions, but a secure line of communication is available to help them coordinate.  Of course, each player could communicate his randomized actions to the other players, but this is an excessive use of communication.  A memoryless channel is a useful way to coordinate their random actions.  Thus, common information is found in Section \ref{section game theory} to be the significant quantity in this situation.

\section{Preliminaries and Problem Definition}
\label{section preliminaries}

\subsection{Notation}
\label{subsection notation}

We represent random variables as capital letters, $X$, and their alphabets are written in script, ${\cal X}$.  Sequences, $X_1,...,X_n$ are indicated with a superscript $X^n$.  Distribution functions, $p_X(x)$, are usually abbreviated as $p(x)$ when there is no confusion.

Accented variables, $\hat{X}$, indicate different variables for each accent, but their alphabets are all the same, ${\cal X}$.  Similarly, distribution functions written with an accent or different letter, such as $p(x)$ versus $\hat{p}(x)$, represent different distributions.

Markov chains, satisfying $p(x,y,z) = p(x,y)p(z|y)$, are represented with dashes, $X-Y-Z$.

\noindent
(Wyner's) common information:
$$C(X;Y) \triangleq \min_{X-U-Y} I(X,Y;U).$$
Conditional common information:
$$C(X;Y|W) \triangleq \min_{X-(U,W)-Y} I(X,Y;U|W).$$
Total variation distance:
$$\left| \left| p - q \right| \right|_1 \triangleq \frac{1}{2} \sum_x |p(x) - q(x)|.$$

\subsection{Problem Specific Definitions}
\label{subsection problem specific definitions}

A source $X^n$ is distributed i.i.d. according to $\check{p}(x)$.  A description of the source at rate $R_1$ is represented by $I \in \{1,...,2^{nR_1}\}$.  A random variable $J$, uniformly distributed on $\{1,...,2^{nR_2}\}$ and independent of $X^n$, represents the common random bits at rate $R_2$ known at both the encoder and decoder.  The decoder generates a channel output $Y^n$ based only on $I$ and $J$.

The channel being simulated has a the conditional distribution $q(y|x)$, thus the {\em desired joint distribution} is $\check{p}(x)q(y|x)$.

\begin{definition}
A $(2^{nR_1}, 2^{nR_2}, n)$ {\em channel simulation code} consists of a randomized encoding function,
$$F_n: {\cal X}^n\times \{1,2,...,2^{nR_2}\} \to \{1,2,...,2^{nR_1}\},$$
and a randomized decoding function,
$$G_n: \{1,2,...,2^{nR_1}\} \times \{1,2,...,2^{nR_2}\} \to {\cal Y}^n.$$
The description $I$ equals $F_n(X^n,J)$, and the channel output $Y^n$ equals $G_n(I,J)$.

Since randomized functions are specified by conditional probability distributions, it is equivalent to say that a $(2^{nR_1}, 2^{nR_2}, n)$ channel simulation code consists of a conditional probability mass function $p(i, y^n|x^n, j)$ with the properties that $p(y^n|i,j,x^n) = p(y^n|i,j)$, $|{\cal I}| = 2^{nR_1}$, and $|{\cal J}| = 2^{nR_2}$.

The {\em induced joint distribution} of a $(2^{nR_1}, 2^{nR_2}, n)$ channel simulation code is the joint distribution on the quadruple $(X^n, Y^n, I, J)$.  In other words, it is the probability mass function,
\begin{eqnarray}
p(x^n, y^n, i, j) & = & p(i, y^n|x^n, j) p(x^n,j), \label{definition induced distribution}
\end{eqnarray}
where $p(x^n,j) = p(j) \prod_{k=1}^n \check{p}(x_k)$ by construction.
\end{definition}

\begin{definition}
A sequence of $(2^{nR_1}, 2^{nR_2}, n)$ channel simulation codes for $n = 1,2,...$ is said to {\em achieve} $q(y|x)$ if the induced joint distributions have marginal distributions $p(x^n,y^n)$ that satisfy
$$\lim_{n \to \infty} \left| \left| p(x^n, y^n) - \prod_{k=1}^n \check{p}(x_k)q(y_k|x_k) \right| \right|_1 = 0.$$
\end{definition}

\begin{definition}
A rate pair $(R_1,R_2)$ is said to be {\em achievable} if there exists a sequence of $(2^{nR_1}, 2^{nR_2}, n)$ channel simulation codes that achieves $q(y|x)$.
\end{definition}

\begin{definition}
The {\em simulation rate region} is the closure of achievable rate pairs $(R_1, R_2)$.
\end{definition}

\section{Main Result}
\label{section main result}

\begin{theorem}
\label{theorem main result}
For an i.i.d. source with distribution $\check{p}(x)$ and a desired memoryless channel with conditional distribution $q(y|x)$, the simulation rate region is the set,
\begin{eqnarray}
S \triangleq \left\{ (R_1,R_2) \in {\cal R}^2: \right. & & \exists p(x,y,u) \in 
D \mbox{ s.t.} \nonumber \\
R_1 & \geq & I(X;U), \nonumber \\
R_1 + R_2 & \geq & \left. I(X,Y;U) \right \}, \label{definition S}
\end{eqnarray}
where
\begin{eqnarray}
D \triangleq \{p(x,y,u): & & (X,Y) \sim \check{p}(x)q(y|x), \nonumber \\
& & X - U - Y \mbox{ form a Markov chain}, \nonumber \\
& & |{\cal U}| \leq |{\cal X}| |{\cal Y}| + 1\}. \label{definition D}
\end{eqnarray}

\end{theorem}

\section{Observations and Examples}
\label{section examples}

Two extreme points of the simulation rate region $S$ fall directly from its definition.  If $R_2 = 0$, the second inequality in (\ref{definition S}) dominates.  Thus, the minimum rate $R_1$ is the common information $C(X;Y)$.  This coincides with the intuition provided by Wyner's result in \cite{Wyner}.  At the other extreme, using the data processing inequality on the first inequality of (\ref{definition S}) yields $R_1 \geq I(X;Y)$ no matter how much common randomness is available, and this is achieved when $R_2 \geq H(Y|X)$.\footnote{$R_2$ doesn't necessary have to be as large as $H(Y|X)$ for $(I(X;Y), R_2)$ to be in the simulation rate region.}  Source coding results and the coordinated action work of Cover and Permuter in \cite{Permuter} illustrate that with a description rate of $I(X;Y)$ we can create a codebook of output sequences in such a way that we'll likely be able to find a jointly typical output sequence for each input sequence from the source.  Consequently, we can then randomize the codebook using common randomness to actually simulate the channel, as Bennett and Shor proved in \cite{Bennett}.

\subsection{Binary Erasure Channel}
\label{subsection erasure channel}

For a Bernoulli-half source $X$, let us demonstrate the simulation rate region for the binary erasure channel.  $Y$ is an erasure with probability $P_e$ and is equal to $X$ otherwise.  The distributions in $D$ that produce the boundary of the simulation rate region are formed by cascading two binary erasure channels as shown in Figure \ref{figure erasure test channel}, where
\begin{eqnarray}
p_2 & \in & \left[ 0, \min \left\{ \frac{1}{2}, P_e \right\} \right], \nonumber \\
p_1 & = & 1 - \frac{1-P_e}{1-p_2}. \nonumber
\end{eqnarray}
The mutual information terms in (\ref{definition S}) become
\begin{eqnarray}
I(X;U) & = & 1 - p_1, \nonumber \\
I(X,Y;U) & = & h(P_e) + (1-p_1) (1 - h(p_2)), \nonumber
\end{eqnarray}
where $h$ is the binary entropy function.

Figure \ref{figure BEC 075 tradeoff} shows the boundary of the simulation rate region for erasure probability $P_e = 0.75$.  The required description rate $R_1$ varies from $C(X;Y) = h(0.75) = 0.811$ bits to $I(X;Y) = 0.25$ bits as the rate of common randomness runs between $0$ and $H(Y|X) = h(0.75) = 0.811$ bits.

\begin{figure}
\psfrag{l1}[][][1.0]{$X$}
\psfrag{l2}[][][1.0]{$U$}
\psfrag{l3}[][][1.0]{$Y$}
\psfrag{l4}[][][0.8]{$0$}
\psfrag{l5}[][][0.8]{$1$}
\psfrag{l6}[][][0.8]{$p_1$}
\psfrag{l7}[][][0.8]{$p_1$}
\psfrag{l8}[][][0.8]{$p_2$}
\psfrag{l9}[][][0.8]{$p_2$}
\psfrag{l10}[][][0.8]{$0$}
\psfrag{l11}[][][0.8]{$e$}
\psfrag{l12}[][][0.8]{$1$}
\centering
\includegraphics[width=2.5in]{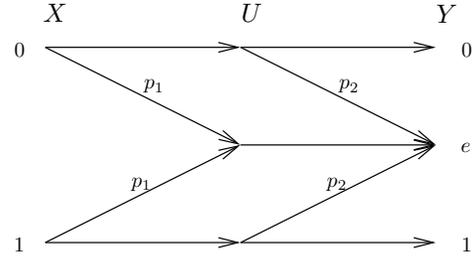}
\caption{The Markov chains $X-U-Y$ that give the boundary of the simulation rate region for the binary erasure channel with a Bernoulli-half input are formed by cascading two erasure channels.}
\label{figure erasure test channel}
\end{figure}

\begin{figure}
\psfrag{BEC}[][][0.9]{BEC Simulation Rate Region, $P_e = 0.75$}
\psfrag{l1}[][][0.8]{$I(X;Y)$}
\psfrag{l2}[][][0.8]{$C(X;Y)$}
\centering
\includegraphics[width=3in]{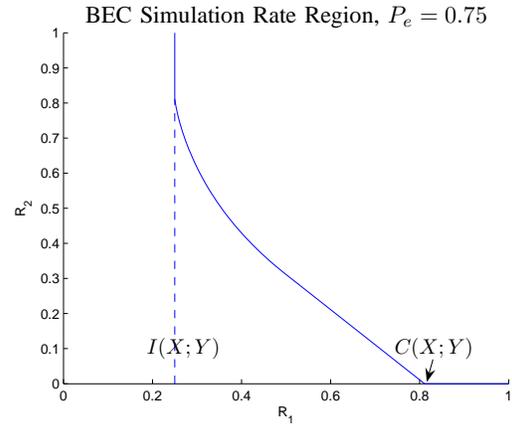}
\caption{Boundary of the simulation rate region for a binary erasure channel with erasure probability $P_e = 0.75$ and a Bernoulli-half input, where $R_1$ is the description rate and $R_2$ is the rate of common randomness.  Without common randomness, a description rate of $C(X;Y)$ is required to simulate the channel.  With unlimited common randomness, a description rate of $I(X;Y)$ suffices.}
\label{figure BEC 075 tradeoff}
\end{figure}

\section{Sketch of Converse}
\label{section converse}

Let $(R_1, R_2)$ be an achievable rate pair.  Then for each $\epsilon \in (0,1/4)$ there exists a $(2^{nR_1}, 2^{nR_2}, n)$ channel simulation code with an induced joint distribution $p(x^n,y^n,i,j)$ such that
$$\left| \left| p(x^n, y^n) - \prod_{k=1}^n \check{p}(x_k)q(y_k|x_k) \right| \right|_1 < \epsilon.$$
Let the random variable $K$ be uniformly distributed over the set $\{1,...,n\}$.  The variable $K$ will serve as a random time index.

\subsection{Entropy Bounds}
\label{subsection entropy bounds}

The joint distribution of the sequences $(X^n,Y^n)$ is close in total variation to an i.i.d. distribution, so we can extend Lemma 2.7 of \cite{Csiszar} to obtain two bounds:
\begin{eqnarray}
\left| H(X^n,Y^n) - \sum_{k=1}^n H(X_k,Y_k) \right| & \leq & n g(\epsilon), \label{equation entropy bound} \\
I(X_K,Y_K;K) & \leq & n g(\epsilon), \label{equation information bound}
\end{eqnarray}
where
\begin{equation}
\label{definition g epsilon}
g(\epsilon) \triangleq 4\epsilon \left( \log|{\cal X}| + \log|{\cal Y}| + \log \frac{1}{\epsilon} \right).
\end{equation}
Notice that $\lim_{\epsilon \downarrow 0} g(\epsilon) = 0$.

\subsection{Epsilon Rate Region}
\label{subsection epsilon rate region}

Define an epsilon rate region,
\begin{eqnarray}
S_{\epsilon} \triangleq \left\{ (R_1,R_2) \in {\cal R}^2 \right. & : & \exists p(x,y,u) \in 
D_{\epsilon} \mbox{ s.t.} \nonumber \\
R_1 & \geq & I(X;U) - 2g(\epsilon), \nonumber \\
R_1 + R_2 & \geq & \left. I(X,Y;U) - 2g(\epsilon) \right \}, \nonumber
\end{eqnarray}
where
\begin{eqnarray}
D_{\epsilon} \triangleq \{p(x,y,u) & : & ||p(x,y) - \check{p}(x)q(y|x)||_1 < \epsilon, \nonumber \\
& & X - U - Y \mbox{ form a Markov chain}, \nonumber \\
& & |{\cal U}| \leq |{\cal X}| |{\cal Y}| + 1\}. \label{definition D epsilon}
\end{eqnarray}

\begin{lemma}
\label{lemma epsilon rate region}
$$(R_1, R_2) \in S_{\epsilon}.$$
\end{lemma}

\begin{proof}
We use familiar information theoretic inequalities, and the fact that $X^n$ and $J$ are independent, to bound $R_1$ and the sum rate $R_1 + R_2$.
\begin{eqnarray}
nR_1 & \geq & H(I) \nonumber \\
& \geq & H(I|J) \nonumber \\
& \geq & I(X^n;I|J) \nonumber \\
& = & I(X^n;I,J). \label{equation Ixij}\\
n(R_1 + R_2) & \geq & H(I,J) \nonumber \\
& \geq & I(X^n,Y^n;I,J). \label{equation Ixyij}
\end{eqnarray}
We then lower bound the r.h.s. of (\ref{equation Ixij}) and (\ref{equation Ixyij}) using similar steps.  Here we proceed from (\ref{equation Ixyij}).
\begin{eqnarray*}
I(X^n;Y^n;I,J) & = & H(X^n,Y^n) - H(X^n,Y^n|I,J) \\
& \geq & H(X^n,Y^n) - \sum_{k=1}^n H(X_k,Y_k|I,J) \\
& \geq & \sum_{k=1}^n I(X_k,Y_k;I,J) - ng(\epsilon) \\
& = & nI(X_K,Y_K;I,J|K) - ng(\epsilon) \\
& \geq & nI(X_K,Y_K;I,J,K) - 2n g(\epsilon).
\end{eqnarray*}
The second inequality comes from (\ref{equation entropy bound}), and the last inequality comes from (\ref{equation information bound}).

The joint distribution of the pair $(X_K,Y_K)$ can be shown to satisfy the total variation constraint in (\ref{definition D epsilon}).  Finally, we acknowledge the Markovity of the triple $X_K - (I,J,K) - Y_K$  to complete the proof of the lemma.  (The cardinality bound of $U$ in (\ref{definition D epsilon}) is shown to be satisfiable via a generalized Caratheodory theorem.)
\end{proof}

\subsection{Lower semi-continuity}
\label{subsection lower semi continuity}
The epsilon rate regions decrease to the simulation rate region as epsilon decreases to zero.

\begin{lemma}
\label{lemma lower semi continuity}
$$\bigcap_{\epsilon \in (0,1/2)} S_\epsilon \subset S.$$
\end{lemma}

\section{Sketch of Achievability}
\label{section achievability}

\subsection{Resolvability}
\label{subsection resolvability}

One key tool for the achievability proof is summarized in Lemma \ref{lemma quantized input}.  This lemma is implied by the resolvability work of Han and Verd\'{u} in \cite{Han}, but the concept was first introduced by Wyner in Theorem 6.3 of \cite{Wyner}.

\begin{lemma}
\label{lemma quantized input}
For any discrete distribution $p(u,v)$ and each $n$, let ${\cal C}^{(n)} = \{U^n(m)\}_{m=1}^{2^{nR}}$ be a ``codebook'' of sequences each independently drawn according to $\prod_{k=1}^n p_U(u_k)$.

For a fixed codebook, define the distribution
$$Q(v^n) = 2^{-nR} \sum_{m=1}^{2^{nR}} \prod_{k=1}^n p_{V|U}(v_k|U_k(m)).$$

Then if $R > I(V;U)$,
$$\lim_{n \to \infty} {\mathbb E} \left| \left| Q(v^n) - \prod_{k=1}^n p_V(v_k) \right| \right|_1 = 0,$$
where the expectation is with respect to the randomly constructed codebooks ${\cal C}^{(n)}$.
\end{lemma}

\subsection{Existence of Achievable Codes}
\label{subsection existence of achievable codes}

Assume that $(R_1,R_2)$ is in the interior of $S$.  Then there exists a distribution $p^{\ast}(x,y,u) \in D$ such that $R_1 > I(X;U)$ and $R_1 + R_2 > I(X,Y;U)$.

For each $n$, let $(I,J)$ be uniformly distributed on $\{1,...,2{nR_1}\} \times \{1,...,2^{nR_2}\}$.  We apply Lemma \ref{lemma quantized input} twice, once with $V=(X,Y)$ and again with $V=X$, to assert that there exists a sequence of ``codebooks'' ${\cal C}^{(n)} = \{U^n(i,j)\}_{(i,j) \in {\cal I}\times{\cal J}}$, $n=1,2,...$ with the properties
\begin{eqnarray}
\lim_{n \to \infty} \left| \left| Q(x^n,y^n) - \prod_{k=1}^n p_{X,Y}^{\ast}(x_k,y_k) \right| \right|_1 & = & 0, \;\;\;\label{equation close joint distribution}\\
\lim_{n \to \infty} \left| \left| Q(x^n,j) - p(j) \prod_{k=1}^n p_{X}^{\ast}(x_k) \right| \right|_1 & = & 0, \;\;\; \label{equation close to independent}
\end{eqnarray}
where $Q(x^n,y^n)$ and $Q(x^n,j)$ are marginal distributions derived from the joint distribution
\begin{eqnarray*}
Q(x^n,y^n,i,j) & = & p(i,j) \prod_{k=1}^n p_{X,Y|U}^{\ast} (x_k,y_k|U_k(i,j)).
\end{eqnarray*}

In an indirect way, we've constructed a sequence of joint distributions $Q(x^n,y^n,i,j)$ from which we can derive channel simulation codes that achieve $q(y|x)$.  The Markovity of $p^{\ast}$ implies the Markov property $Q(x^n,y^n|i,j) = Q(x^n|i,j)Q(y^n|i,j)$.  Let
\begin{eqnarray*}
\hat{p}(i|x^n,j) & = & Q(i|x^n,j), \\
\hat{p}(y^n|i,j) & = & Q(y^n|i,j).
\end{eqnarray*}
Considering (\ref{equation close joint distribution}) and (\ref{equation close to independent}) with the properties of total variation and $p^{\ast}$ in mind, it can be shown that $\hat{p}(i,y^n|x^n,j) = \hat{p}(i|x^n,j)\hat{p}(y^n|i,j)$ is a sequence of channel simulation codes that achieves $q(y|x)$.

\subsection{Comment on Achievability Scheme}
\label{subsection achievability comment}

This channel simulation scheme requires randomization at both the encoder and decoder.  In essence, a codebook of independently drawn $U^n$ sequences is overpopulated so that the encoder can choose one randomly from many that are jointly typical with $X^n$.  The decoder then randomly generates $Y^n$ conditioned on $U^n$.

\section{Game Theory}
\label{section game theory}

Our framework finds motivation in a game theoretic setting.  Consider a zero-sum repeated game between two teams.  Team A consists of two players who on the $i$th iteration take actions $X_i \in {\cal X}$ and $Y_i \in {\cal Y}$.  The opponents on team B take combined action $Z_i \in {\cal Z}$.  All action spaces ${\cal X, Y}$, and ${\cal Z}$ are finite.  The payoff for team A at each iteration is a time-invariant finite function $\Pi(X_i, Y_i, Z_i)$ and is the loss for team B.  Each team wishes to maximize its time-averaged expected payoff.

Assume that team A plays conservatively, attempting to maximize the expected payoff for the worst-case actions of team B.  Then the payoff at the $i$th iteration is
\begin{eqnarray}
\label{loss function}
\Theta_i & \triangleq & \min_{z \in {\cal Z}} {\mathbb E} \left[ \Pi(X_i,Y_i,z)|X^{i-1},Y^{i-1} \right].
\end{eqnarray}
Clearly, (\ref{loss function}) could be maximized by finding an optimal mixed strategy $p^{\ast}(x,y)$ that maximizes $\min_{z \in {\cal Z}} {\mathbb E} \left[ \Pi(X,Y,z) \right]$ and choosing independent actions each iteration.  This would correspond to the minimax strategy.  However, now we introduce a new constraint:  The players on team A have a limited secure channel of communication.  Player 1, who chooses the actions $X^n$, communicates at rate $R$ to Player 2, who chooses $Y^n$.

Let $U$ be the message passed from Player 1 to Player 2.  We say a rate $R$ is achievable for payoff $\Theta$ if there exists a sequence of random variable triples $(X^n,Y^n,U)$ that each form Markov chains \footnote{This Markov chain requirement can be relaxed to the more physically relevant requirement that $X_k - (U, X^{k-1}, Y^{k-1}) - Y_k$ for all $k$.} $X^n - U - Y^n$ and such that
$|{\cal U}| \leq 2^{nR}$ and
\begin{eqnarray}
\label{equation game achievability}
\lim_{n \to \infty} {\mathbb E} \left[ \frac{1}{n} \sum_{i=1}^n \Theta_i \right] & > & \Theta.
\end{eqnarray}

Let $R(\Theta)$ be the infimum of achievable rates for payoff $\Theta$.  We claim that $R(\Theta)$ is the least average common information of all combinations of strategies that achieve average payoff $\Theta$.  Define,
\begin{eqnarray*}
R_0(\Theta) & \triangleq & \min C(X;Y|W) \\
& & \mbox{s.t. } {\mathbb E} \left[ \min_{z \in {\cal Z}} {\mathbb E} \left[ \Pi(X,Y,z)|W \right] \right] \geq \Theta.
\end{eqnarray*}

\begin{theorem}
\label{theorem game theory}
$$R(\Theta) = R_0(\Theta).$$
\end{theorem}

{\em Converse Sketch:} \\
The important elements of the converse are the inequalities
\begin{eqnarray*}
n(R(\Theta) + \epsilon) & > & H(U) \\
& \geq & I(X^n,Y^n;U) \\
& = & \sum_{i=1}^n I(X_i,Y_i;U|X^{i-1},Y^{i-1}) \\
& = & nI(X_K,Y_K;U|X^{K-1},Y^{K-1},K),
\end{eqnarray*}
for all $\epsilon > 0$, where $K$ is uniformly distributed on $\{1,...,n\}$.  Now identify the tuple $(X^{K-1},Y^{K-1},K)$ as the auxiliary random variable $W$.

{\em Achievability Comment:} \\
The random variable $W$ serves as a time sharing variable to combine strategies of high and low correlation.

\section{Acknowledgment}

The author would like to thank his advisor, Tom Cover, for encouraging the study of coordination via communication, and Young-Han Kim for his suggestions and encouragement.  This work is supported by the National Science Foundation through grants CCF-0515303 and CCF-0635318.

\end{document}